\documentclass[a4paper,english]{lipics-v2021}

\usepackage{microtype} 
\usepackage{bm}
\usepackage{enumitem}
\usepackage{amsthm}
\usepackage{thmtools}

\newtheorem{thm}{Theorem}
\newtheorem{lem}[thm]{Lemma}
\newtheorem{cor}[thm]{Corollary}

\title{The Shortest Path with Increasing Chords in a Simple Polygon}
\author{Mart Hagedoorn}{TU Dortmund, Germany}{mart.hagedoorn@tu-dortmund.de}{}{}
\author{Irina Kostitsyna}{TU Eindhoven, The Netherlands}{i.kostitsyna@tue.nl}{}{}

\authorrunning{M. Hagedoorn and I. Kostitsyna} 

\Copyright{Mart Hagedoorn and Irina Kostitsyna} 


\nolinenumbers
\hideLIPIcs

\begin{document}

\maketitle

\begin{abstract}
We study the problem of finding the shortest path with \emph{increasing chords} in a simple polygon.
A path has increasing chords if and only if for any points $a, b, c$, and $d$ that lie on the path in that order, $|ad| \geq |bc|$.
In this paper we show that the shortest path with increasing chords is unique and present an algorithm to construct it.
\end{abstract}

\section{Introduction}

\begin{figure}
    \centering
    \includegraphics[page=1]{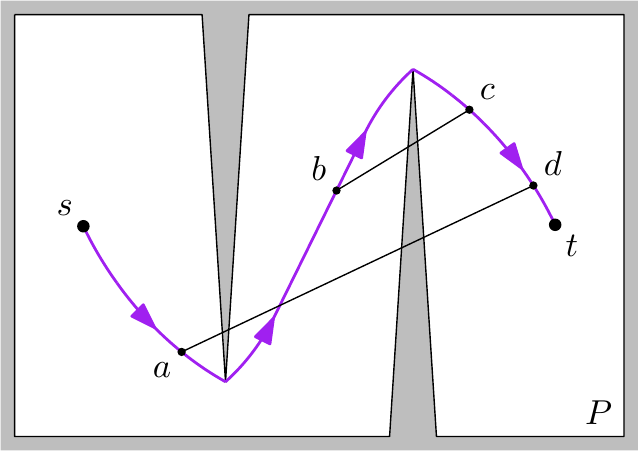}
    \captionof{figure}{The shortest $s$-$t$ path with increasing chords inside simple polygon $P$.}
    \label{fig:increasing_example}
\end{figure}

An $s$-$t$ path $\sigma$ has \emph{increasing chords} if and only if $\sigma$ is a directed path from some point $s$ to $t$ and for any points $a, b, c$, and $d$ that appear in that order on $\sigma$, the Euclidean distance between $a$ and $d$ is greater or equal to the Euclidean distance between $b$ and $c$ \cite{arcs_with_increasing_chords}.
Figure~\ref{fig:increasing_example} shows an example of a shortest $s$-$t$ path with increasing chords in simple polygon $P$.
Furthermore, a path has increasing chords if and only if the path is \emph{self-approaching} in both directions.
A path is self-approaching if, when traversing the path the Euclidean distance to any point on the remainder of the path is not increasing.

Paths with increasing chords are closely related to beacon and greedy routing applications. Path finding with beacon and greedy routing often results in paths that are directed curves such that the distance to the destination is never increasing \cite{beacon-routing, biro2011beacon, greedy_routing_sensor_networks}. 
Paths with this property are called \emph{radially monotone paths}. 
If a path has increasing chords, then every subpath of that path is radially monotone in both directions.

Furthermore, self-approaching paths and paths with increasing chords have the property that the length of the path is bounded in comparison with the Euclidean distance between the start and destination of the path. This bounding factor of paths with increasing chords is only $2\pi / 3$, whereas the bounding factor of self-approaching paths is approximately $5.3331$ \cite{curves_with_increasing, sa}.

The results of this paper are further discussed in the master's thesis of Hagedoorn.

\section{Preliminaries}

An $s$-$t$ path $\zeta$ is a directed curve that starts in point $s$ and ends in point $t$. 
Moreover, path $\zeta$ must lie entirely inside of polygon $P$, i.e. $\zeta \subseteq P$. 
We define $\zeta(p, q)$ to be the subpath of $\zeta$ that lies between points $p$ and $q$.
Paths in a Euclidean space can make smooth turns and sharp turns. In order to differentiate between these types of turns, we use the standard notation of a \emph{bend point} (Fig.~\ref{fig:bend-normal}). A \emph{bend point} $b$ of a piecewise smooth curve $\zeta$ is a point where the first derivative of $\zeta$ is discontinuous.

\begin{figure}
    \centering
    \includegraphics[page=4]{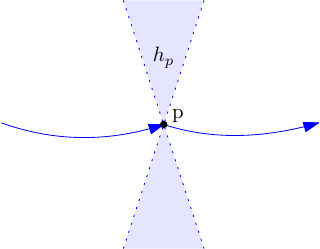}
    \captionof{figure}{The normal $h_p$ of $p$, where $p$ is a bend point.}
    \label{fig:bend-normal}
\end{figure}

Furthermore, we define a \emph{normal} $h_p$ to path $\zeta$ at point $p$ to be the line through $p \in \zeta$ such that $h_p$ is perpendicular to the tangent of $\zeta$ in $p$. If $p$ is a bend point, then we use the normal definition as proposed by Icking et al. \cite{sa}. This definition states that the normal to $\zeta$ at $p$ is the set of lines that are included in the double wedge between the perpendicular lines to the tangents of the smooth pieces meeting at $p$ (Fig. \ref{fig:bend-normal}).

\begin{figure}
    \centering
    \includegraphics[page=2]{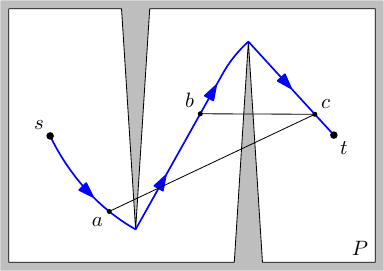}
    \captionof{figure}{The shortest self-approaching $s$-$t$ path inside simple polygon $P$.}
    \label{fig:sa_example}
\end{figure}

\subsection{Self-approaching paths}

Self-approaching paths were first described by Icking et al. \cite{sa}. A path $\pi$ is self-approaching if and only if for any points $a, b$, and $c$ that appear on $\pi$ in that order, the Euclidean distance between $a$ and $c$ is greater or equal to the Euclidean distance between $b$ and $c$ (Fig.~\ref{fig:sa_example}). Furthermore, Icking et al. \cite{sa} showed the following \emph{normal property} of self-approaching paths. 

\begin{lem}[\cite{sa}]\label{lem:sa_normal_prop}
An $s$-$t$ path $\pi$ is self-approaching if and only if the normal to $\pi$ at any point $p \in \pi$ does not intersect the subpath $\pi(p, t)$.
\end{lem}

\begin{figure}
    \centering
    \includegraphics[page=3]{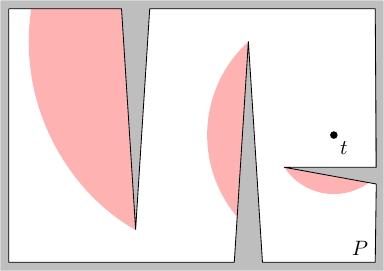}
    \caption{Dead region $\mathcal{D}_t$ (red) for some point $t$ in simple polygon $P$.}
    \label{fig:dead_regions}
\end{figure}

Bose et al. \cite{self-approaching_simple_polygons} proposed an algorithm for finding the shortest \emph{self-approaching} $s$-$t$ path in a simple polygon $P$ (Fig.~\ref{fig:sa_example}). 
Their algorithm uses the notion of \emph{dead regions}. A dead region $\mathcal{D}_t$ for point $t$ is a set of points such that for any point $s \in \mathcal{D}_t$, no self-approaching $s$-$t$ path exists (Fig. \ref{fig:dead_regions}). They proved that the shortest self-approaching $s$-$t$ path $\pi$ is the $s$-$t$ geodesic in $P \setminus \mathcal{D}_t$, i.e.\ the set difference of the polygon and the dead region for point $t$.
Therefore, the shortest self-approaching $s$-$t$ path consists of straight line segments and segments that are boundaries of $\mathcal{D}_t$.
Moreover, they showed that if $v$ is a point on $\pi$ and $v$ lies on a boundary of $\mathcal{D}_t$, then the normal $h_v$ to $\pi$ at $v$ must \emph{touch} the subpath $\pi(v, t)$. 
A line $\ell$ touches a curve $\gamma$ if there exists some point $p$ such that $\ell \cap \gamma = \{p\}$ and for the normal $h_p$ to $\gamma$ at point $p$, $\ell \in h_p$ or $\ell = h_p$ when $p$ is a bend point or not, respectively.

The equations that define boundaries of $\mathcal{D}_t$ are transcendental equations and likely cannot be solved or evaluated analytically \cite{self-approaching_simple_polygons}. Therefore, the shortest self-approaching path can only be found if we assume these equations can be solved or that an approximation of the path will be calculated.

\subsection{Paths with increasing chords}

As mentioned before, a path $\sigma$ has increasing chords if and only if $\sigma$ is self-approaching in both directions. Furthermore, for paths with increasing chords we can again define the normal property:

\begin{lem}\label{lem:ic_normal_prop}
An $s$-$t$ path $\sigma$ has increasing chords if and only if the normal to $\sigma$ at any point $p \in \sigma$ does not intersect the subpaths $\sigma(s, p)$ and $\sigma(p, t)$.
\end{lem}

\begin{proof}
A path $\sigma$ with increasing chords is self-approaching from $s$ to $t$ and from $t$ to $s$.
Therefore, at any point $p \in \sigma$ the normal properties of the self-approaching $s$-$t$ and $t$-$s$ paths state that $\sigma(s, p)$ and $\sigma(p, t)$ cannot cross the normal to $\sigma$ at point $p$.
\end{proof}

\begin{figure}
  \centering
  \includegraphics[page=3]{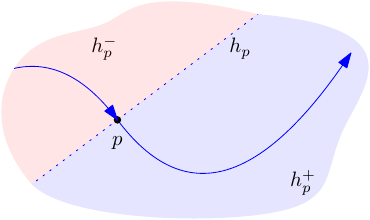}
  \captionof{figure}{Example figure showing the normal $h_p$ at point $p$ and the corresponding half planes $h_p^-$ and $h_p^+$.}
  \label{fig:half_planes}
\end{figure}

This property above can be reformulated in terms of the \emph{negative} and \emph{positive half-plane}. Using the definition from Bose et al. \cite{self-approaching_simple_polygons}, the positive half-plane $h_p^+$ of path $\zeta$ at point $p$ is the closed half-plane that is defined by normal $h_p$ to $\zeta$ and contains 
at point $p$ which is congruent to the direction of $\zeta$ at $p$. Analogously, the negative half-plane $h_p^-$ is congruent with the opposite direction of $\zeta$ at $p$ (Fig.~\ref{fig:half_planes}).

\begin{cor} \label{lem:half_plane}
An $s$-$t$ path $\sigma$ has increasing chords if and only if, for any line $h$ normal to $\sigma$ at point $p \in \sigma$, the subpath $\sigma(s, p)$ lies completely in the negative half-plane $h_p^-$ and the subpath $\sigma(p, t)$ lies completely in the positive half-plane $h_p^+$.
\end{cor}

\section{Shortest Path with Increasing Chords}

In this section we first show that a shortest $s$-$t$ path with increasing chords in a simple polygon is unique. Therefore, we only need to search for one shortest $s$-$t$ path with increasing chords.
The following two proofs are extensions from the analogous propositions about self-approaching paths given by Bose et al.\ \cite{self-approaching_simple_polygons}.

\begin{figure}
    \centering
    \includegraphics[page=30]{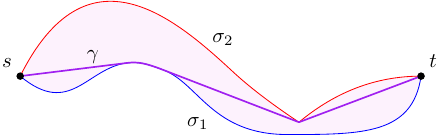}
    \caption{Geodesic $\gamma$ (purple) that lies between paths $\sigma_1$ (blue) and $\sigma_2$ (red) with increasing chords.}
    \label{fig:between}
\end{figure}

\begin{restatable}{lem}{lempathbetween}\label{geodesic}
A geodesic path $\gamma$ between two distinct paths with increasing chords $\sigma_1$ and $\sigma_2$ also has increasing chords (Fig.~\ref{fig:between}).
\end{restatable}

\begin{proof}
We will use the fact that the geodesic must lie in both convex hulls of both paths, i.e.\ $\gamma \subseteq \textit{CH}(\sigma_1)$ and $\gamma \subseteq \textit{CH}(\sigma_2)$.
Furthermore, any point $p \in \gamma$ either lies on one of the paths $\sigma_1$ or $\sigma_2$ or on a straight line that is bitangent to $\sigma_1$, $\sigma_2$ or both $\sigma_1$ and $\sigma_2$.

First, we consider the case that $p$ is not a bend point but lies on a smooth section of $\sigma_1$ or $\sigma_2$.
Let, w.l.o.g., $p \in \sigma_1$. The positive half-plane $h_p^+$ of the normal to $\sigma_1$ at $p$ contains subpath $\sigma_1(p, t)$ and the negative half-plane $h_p^-$ of the normal at $p$ contains subpath $\sigma_1(s, p)$. Therefore, $h_p^+$ also contains the convex hull of $\sigma_1(p, t)$ and $h_p^-$ contains the convex hull of $\sigma_1(s, p)$. Hence, $h_p^+$ contains the subpath $\gamma(p, t)$ and $h_p^-$ contains the subpath $\gamma(s, p)$ of the geodesic.

Now let $p$ be a bend point lying on $\sigma_1$. Subpath $\sigma_1(p, t)$ must be fully contained in the two positive half-planes, which are defined by the normals of the smooth pieces of $\sigma_1$ meeting at $p$. Analogously $\sigma_1(s, p)$ is fully contained in the two negative half-planes.
The two normals of the geodesic path at this point must also lie in between the two normals to the boundary path. 
Thus, the intersection of the two positive half-planes of the normals to the geodesic contains the convex hull of the subpath $\sigma_1(p, t)$, and, therefore, the rest of the geodesic path $\gamma(p, t)$. Analogously, the intersection of the two negative half-planes of the normals to the geodesic contains the convex hull of the subpath $\sigma_1(s, p)$, and, therefore, the geodesic subpath $\gamma(s, p)$.

Finally, let $p$ lie on a segment of $\gamma$ that is bitangent to $\sigma_1$, $\sigma_2$ or both $\sigma_1$ and $\sigma_2$. Let us consider one of the end-points $p^*$ of this bitangent. Assume, that $p^* \in \gamma(p, t)$. The normal to $\gamma$ at $p$ is parallel to the normal to $\gamma$ at $p^*$. By one of the cases considered above, the positive half-plane at $p^*$ will contain $\gamma(p^*, t)$, and, therefore, the positive half-plane of the normal to $\gamma$ at $p$ will also contain the subpath $\gamma(p, t)$. If $p^* \in \gamma(s, p)$, then we can follow analogues steps to show that $\gamma(s, p)$ lies in the negative half-plane of $p$.

Thus, $\gamma$ has increasing chords.
\end{proof}

Using Lemma~\ref{geodesic} we can prove the following Theorem.

\begin{thm}\label{unique}
A shortest $s$-$t$ path with increasing chords in a simple polygon is unique.
\end{thm}

\begin{proof}
Assume there exist two distinct shortest $s$-$t$ paths with increasing chords $\sigma_1$ and $\sigma_2$ in a polygon $P$. Then, by Lemma~\ref{geodesic}, there exists a shorter path that is a geodesic path enclosed between $\sigma_1$ and $\sigma_2$.
\end{proof}

In the next theorem we show that if we subtract the union of dead regions $\mathcal{D}_t$ and $\mathcal{D}_s$ from a simple polygon $P$, then the shortest $s$-$t$ path in this space is also the shortest $s$-$t$ path with increasing chords.
This theorem is proven by contradiction. 
In more detail, we show that if there is a point where the normal property is violated, then no self-approaching path can exist.

\begin{thm}\label{thm:geodesic_dead}
Let $s$ and $t$ be two points in a simple polygon $P$. The shortest path between $s$ and $t$ in $P \setminus (\mathcal{D}_t \cup \mathcal{D}_s)$ is the shortest $s$-$t$ path with increasing chords in simple polygon $P$.
\end{thm}

\begin{proof}
\begin{figure}
\centering
\begin{minipage}[t]{.45\textwidth}
    \centering
    \includegraphics[page=31]{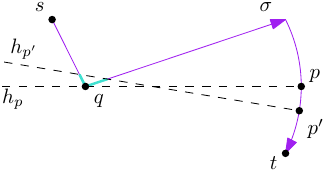}
    \caption{Shortest path $\sigma$ (purple) where the normal property does not hold in $p'$, normal $h_p$ touches point $q$ that lies on subpath $\rho$ (blue) of $\sigma$.}
    \label{fig:finding_q}
\end{minipage}%
\hfill
\begin{minipage}[t]{.45\textwidth}
    \centering
    \includegraphics[page=32]{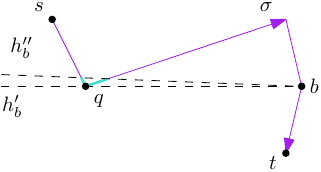}
    \caption{Shortest path $\sigma$ (purple) where the normal property does not hold in $p'$, normal $h'_b$ touches point $q$ that lies on subpath $\rho$ (blue) of $\sigma$.}
    \label{fig:finding_q_with_b}
\end{minipage}
\end{figure}

Assume that $\sigma$ does not have increasing chords, hence the normal property is violated at some point of $\sigma$.
Let point $p \in \sigma$ be the last point on $\sigma$ for which the normal property holds. Point $p$ is not guaranteed to exist. However, if $p$ exists, then there also exists point $p' \in \sigma(p, t)$ such that $p'$ lies in the $\epsilon$-neighborhood of $p$ for arbitrary small $\epsilon$ and the normal property does not hold in $p'$ (Fig.~\ref{fig:finding_q}).
Since in $p'$ the normal property does not hold, there is some subpath $\rho$ of $\sigma(s, p')$ or $\sigma(p', t)$ that lies in the positive or negative half-plane defined by the normal through $h_p'$, respectively.
Furthermore, the normal $h_p$ touches $\sigma$ at some point $q \in \rho$.
However, if $p$ does not exist, there exists a bend point $b$ such that there are two lines $h'_b, h''_b \in h_b$ among the set of lines in the normal $h_b$, such that $h'_b$ touches $\sigma$ at some point $q \in \sigma$ and $h''_b$ intersects $\sigma$ (Fig.~\ref{fig:finding_q_with_b}). All cases where $p$ does not exist are simply analogous to the cases where $p$ exists. Therefore, we will only cover the cases where $p$ exists.

We assume, without loss of generality, that line segment $\overline{pq}$ is horizontal, $p$ lies to the right of $q$, and $\sigma(s, p)$ lies above $h_p$.
Let $e$ be the segment of $\sigma$ containing $p$ and $f$ be the segment of $\sigma$ containing $q$.
We must consider the cases where $q \in \sigma(s, p)$ and $q \in \sigma(p, t)$.
Furthermore, since $\sigma$ is a geodesic in $P \setminus (\mathcal{D}_t \cup \mathcal{D}_s)$ the segments $e$ and $f$ can be straight line segments, or boundaries of a dead region of $\mathcal{D}_t$ or $\mathcal{D}_s$.
If $f$ is a straight line segment, then $q$ must be an end point of $f$ and thus be a vertex of polygon $P$.
Here, we will only cover the case where $q \in \sigma(s, p)$, $e$ is a boundary of $\mathcal{D}_t$, and $q$ is a vertex of $P$.
For each possible combination of the segment types of $e$ and $f$ and the location of $q$ on the path $\sigma$ we show that these cases are not possible by contradiction.

Once we have shown that no point exists where the normal property is violated, $\sigma$ indeed has increasing chords by Lemma~\ref{lem:sa_normal_prop}.
Furthermore, $\sigma$ must be the shortest path with increasing chords. Any path that is shorter than $\sigma$ must go through either of the dead regions $\mathcal{D}_t$ or $\mathcal{D}_s$. 
Therefore, any path shorter than $\sigma$ cannot have increasing chords and $\sigma$ must be the shortest path with increasing chords.

\begin{figure}
\centering
\begin{subfigure}[t]{.45\textwidth}
  \centering
  \includegraphics[page=1]{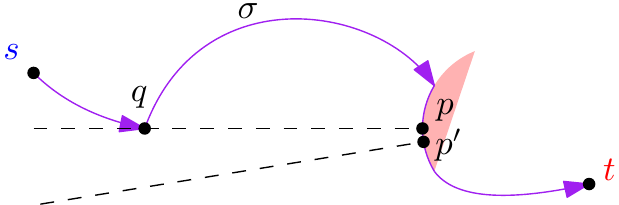}
  \caption{The point of curvature of $\sigma$ at point $p$ lies to the right of $p$.}
  \label{fig:curvature_left}
\end{subfigure}%
\hfill
\begin{subfigure}[t]{.45\textwidth}
  \centering
  \includegraphics[page=2]{figures/geodesic_dead_regions.pdf}
  \caption{The point of curvature $v$ of $\sigma$ at point $p$ lies to the left of $p$. The green path is the shortest self-approaching $p$-$t$ path.}
  \label{fig:q_straight_p_dt_right}
\end{subfigure}
\\
\bigskip
\begin{subfigure}[t]{.6\textwidth}
  \centering
  \includegraphics[page=22]{figures/geodesic_dead_regions.pdf}
  \caption{The point of curvature $v$ of $\sigma$ at point $p$ lies to the left of $p$. The orange path is the shortest self-approaching $v$-$s$ path, that intersects with $h_p$ in $w$ and $w'$.}
  \label{fig:q_straight_p_dt_or}
\end{subfigure}%
\caption{Geodesic $\sigma$ (purple) where point $q \in \sigma(s, p)$, $p$ lies on a boundary of $\mathcal{D}_t$ (red), and $q$ is a vertex of $P$.}
\label{fig:q_straight_p_dt}
\end{figure}
\paragraph*{Point \bm{$q \in \sigma(s, p)$}, \bm{$e$} is a boundary of \bm{$\mathcal{D}_t$}, and \bm{$q$} is a vertex of \bm{$P$} (Fig.~\ref{fig:q_straight_p_dt})}

\noindent The center of curvature of $\sigma$ at $p$ must lie to the left of $p$. Otherwise the normal $h_{p'}$ to $\sigma$ cannot intersect $\sigma(s, p)$, as is depicted in Fig.~\ref{fig:curvature_left}.
Let $\pi_{pt}$ be the shortest self-approaching path from $p$ to $t$.
Segment $e$ is part of $\mathcal{D}_t$, therefore there must be a point $v \in \pi_{pt}$ touched by normal $h_p$ (Fig.~\ref{fig:q_straight_p_dt_right}).
Because $\pi_{pt}$ goes through $v$, $h_v$ is perpendicular to $h_p$ (if $v$ is a bend point, then there must be a line which is perpendicular $h_p$ in the set of normals $h_v$). 
Furthermore, $\pi_{pt}(v, t)$ lies completely in the positive half-plane $h_v^+$ by the half-plane property.
We will now show that in the region between $\overline{pv}$ and $\pi_{pt}(p, v)$ there are vertices of $P$. 
The straight line segment $\overline{pv}$ concatenated to $\pi_{pt}(v,t)$ would be self-approaching. However, $p'$ lies below $h_p$, thus polygon $P$ must intersect with $\overline{vp}$.
Therefore, the shortest self-approaching $v$-$s$ path $\pi_{vs}$ must first intersect or touch $h_p$ to the right of $v$ at point $w$ and later to the left of $v$ at point $w'$ (Fig.~\ref{fig:q_straight_p_dt_or}). Thus, $|vw'| < |ww'|$ which contradicts the self-approaching property of $\pi_{vs}$.
Hence, this case is not possible if the geodesic between $s$ and $t$ exists.

\begin{figure}
\centering
\begin{minipage}[t]{.45\textwidth}
  \centering
  \includegraphics[page=3]{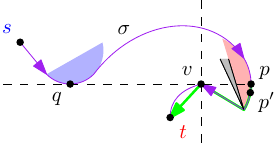}
  \caption{Point $q \in \sigma(s, p)$, $p$ lies on a boundary of $\mathcal{D}_t$ (red), and $q$ lies on a boundary of $\mathcal{D}_s$ (blue). The green path is the shortest self-approaching $p$-$t$ path.}
  \label{fig:q_ds_p_dt}
\end{minipage}%
\hfill
\begin{minipage}[t]{.45\textwidth}
  \centering
  \includegraphics[page=4]{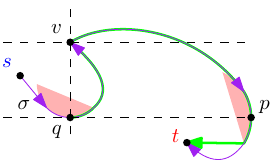}
  \caption{Point $q \in \sigma(s, p)$, and both $p$ and $q$ lie on a boundary of $\mathcal{D}_t$ (red). The green path is the shortest self-approaching $q$-$t$ path.}
  \label{fig:q_dt_p_dt}
\end{minipage}
\end{figure}

\paragraph*{Case 2: \bm{$q \in \sigma(s, p)$}, \bm{$e$} is a boundary of \bm{$\mathcal{D}_t$}, and \bm{$f$} is a boundary of \bm{$\mathcal{D}_s$} (Fig.~\ref{fig:q_ds_p_dt})}

\noindent This case is analogous to Case 1 (Fig.~\ref{fig:q_straight_p_dt}).

\paragraph*{Case 3: \bm{$q \in \sigma(s, p)$}, \bm{$e$} is a boundary of \bm{$\mathcal{D}_t$}, and \bm{$f$} is a boundary of \bm{$\mathcal{D}_t$} (Fig.~\ref{fig:q_dt_p_dt})}

\noindent Since $q$ is touched by $h_p$ and $q$ lies on a boundary of $\mathcal{D}_t$, normal $h_q$ is perpendicular to $h_p$. Furthermore, consider shortest self-approaching $q$-$t$ path $\pi_{qt}$. Because $q$ lies on a boundary of $\mathcal{D}_t$ there is a point $v \in \pi_{qt}$ touched by $h_q$ with $h_v \bot h_q$. Path $\pi_{qt}(p, t)$ or $\pi_{qt}(v, t)$ must lie both in $h_v^+$ and in $h_p^+$, however the intersection of these half-planes is empty; $h_v^+ \cap h_p^+ = \emptyset$. Hence, path $\pi_{qt}$ cannot exist.

\begin{figure}
\centering
\begin{minipage}[t]{.45\textwidth}
  \centering
  \includegraphics[page=5]{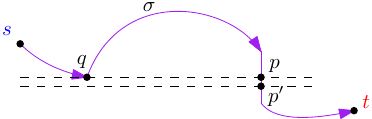}
  \caption{Point $q \in \sigma(s, p)$ and $p$ lies on a straight line segment (orange).}
  \label{fig:p_straigt}
\end{minipage}%
\hfill
\begin{minipage}[t]{.45\textwidth}
  \centering
  \includegraphics[page=6]{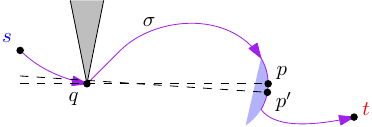}
  \caption{Point $q \in \sigma(s, p)$, $p$ lies on a boundary of $\mathcal{D}_s$ (blue), and $q$ is a vertex of $P$.}
  \label{fig:q_straight_p_ds}
\end{minipage}
\end{figure}

\paragraph*{Case 4: \bm{$q \in \sigma(s, p)$} and \bm{$e$} is a straight line segment (Fig.~\ref{fig:p_straigt})}~

\noindent Since the normals $h_p$ and $h_{p'}$ are parallel it is impossible for $h_{p'}$ to be able to intersect with $\sigma(s, p)$.

\paragraph*{Case 5: \bm{$q \in \sigma(s, p)$}, \bm{$e$} is a boundary of \bm{$\mathcal{D}_s$}, and \bm{$q$} is a vertex of \bm{$P$} (Fig.~\ref{fig:q_straight_p_ds})}~

\noindent Normal $h_p$ touches $\sigma(s, p)$ at point $q$. Point $q$ does not lie in the positive half-plane $h_{p'}^+$ of the shortest self-approaching $p'$-$s$ path $\pi_{p's}$. Since $q$ is a vertex of $P$, $h_{p'}$ intersects with $P$. Therefore, $s$ and $p'$ must lie in different connected components in the intersection of $P$ and the positive half-plane $h_{p'}^+$ of the path $\pi_{p's}$. Hence, no $p'$-$s$ path can exist including $\pi_{p's}$.

\begin{figure}
\centering
\begin{minipage}[t]{.45\textwidth}
  \centering
  \includegraphics[page=7]{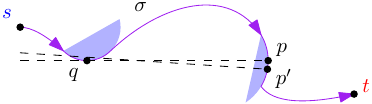}
  \caption{Point $q \in \sigma(s, p)$, and both $p$ and $q$ lie on a boundary of $\mathcal{D}_s$ (blue).}
  \label{fig:q_ds_p_ds}
\end{minipage}%
\hfill
\begin{minipage}[t]{.45\textwidth}
  \centering
  \includegraphics[page=8]{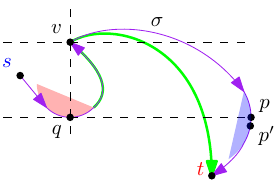}
  \caption{Point $q \in \sigma(s, p)$, $p$ lies on a boundary of $\mathcal{D}_s$ (blue), and $q$ lies on a boundary of $\mathcal{D}_t$ (red). The green path is the shortest self-approaching $q$-$t$ path.}
  \label{fig:q_dt_p_ds}
\end{minipage}
\end{figure}

\paragraph*{Case 6: \bm{$q \in \sigma(s, p)$}, \bm{$e$} is a boundary of \bm{$\mathcal{D}_s$}, and \bm{$q$} is a boundary of \bm{$\mathcal{D}_s$} (Fig.~\ref{fig:q_ds_p_ds}).}~

\noindent Normal $h_p$ touches $\sigma(s, p)$ at point $q$. Point $q$ does not lie in the positive half-plane $h_{p'}^+$ of the shortest self-approaching $p'$-$s$ path $\pi_{p's}$. Since $q$ lies on a boundary of $\mathcal{D}_s$, $h_{p'}$ intersects with this boundary. Therefore, $s$ and $p'$ must lie in different connected components in the intersection of $P \setminus \mathcal{D}_s$ and the positive half-plane $h_{p'}^+$ of the path $\pi_{p's}$. Hence, $\pi_{p's}$ cannot exist.

\paragraph*{Case 7: \bm{$q \in \sigma(s, p)$}, \bm{$e$} is a boundary of \bm{$\mathcal{D}_s$}, and \bm{$f$} is a boundary of \bm{$\mathcal{D}_t$} (Fig.~\ref{fig:q_dt_p_ds})}

\noindent This case is analogous to Case 3 (Fig.~\ref{fig:q_dt_p_dt}).

\begin{figure}
\centering
\begin{minipage}[t]{.45\textwidth}
  \centering
  \includegraphics[page=9]{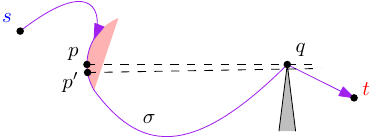}
  \caption{Point $q \in \sigma(p, t)$, $p$ lies on a boundary of $\mathcal{D}_t$ (red), and $q$ is a vertex of $P$.}
  \label{fig:p_dt_q_straight}
\end{minipage}%
\hfill
\begin{minipage}[t]{.45\textwidth}
  \centering
  \includegraphics[page=10]{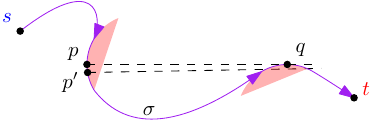}
  \caption{Point $q \in \sigma(p, t)$, and both $p$ and $q$ lie on a boundary of $\mathcal{D}_t$ (red).}
  \label{fig:p_dt_q_dt}
\end{minipage}
\end{figure}

\paragraph*{Case 8: \bm{$q \in \sigma(p, t)$}, \bm{$e$} is a boundary of \bm{$\mathcal{D}_t$}, and \bm{$q$} is a vertex of \bm{$P$} (Fig.~\ref{fig:p_dt_q_straight})}

\noindent This case is analogous to Case 5 (Fig.~\ref{fig:q_straight_p_ds}).

\paragraph*{Case 9: \bm{$q \in \sigma(p, t)$}, \bm{$e$} is a boundary of \bm{$\mathcal{D}_t$}, and \bm{$f$} is a boundary of \bm{$\mathcal{D}_t$} (Fig.~\ref{fig:p_dt_q_dt})}

\noindent This case is analogous to Case 6 (Fig.~\ref{fig:q_ds_p_ds}).

\begin{figure}
\centering
\begin{minipage}[t]{.45\textwidth}
  \centering
  \includegraphics[page=11]{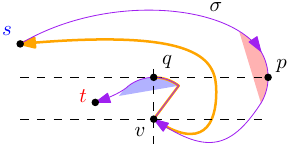}
  \caption{Point $q \in \sigma(p, t)$, $p$ lies on a boundary of $\mathcal{D}_t$ (red), and $q$ lies on a boundary of $\mathcal{D}_s$ (blue). The orange path is the shortest self-approaching $q$-$s$ path.}
  \label{fig:p_dt_q_ds}
\end{minipage}%
\hfill
\begin{minipage}[t]{.45\textwidth}
  \centering
  \includegraphics[page=12]{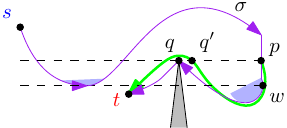}
  \caption{Point $q \in \sigma(p, t)$, both $p$ and $q
  $ lie on a straight line segment (orange), and point $w$ lies on boundary $d$ of $\mathcal{D}_s$ (blue) with $h_p \parallel h_w$. The green path is the shortest self-approaching $p$-$t$ path.}
  \label{fig:p_straight_q_straight}
\end{minipage}
\end{figure}

\paragraph*{Case 10: \bm{$q \in \sigma(p, t)$}, \bm{$e$} is a boundary of \bm{$\mathcal{D}_t$}, and \bm{$f$} is a boundary of \bm{$\mathcal{D}_s$} (Fig.~\ref{fig:p_dt_q_ds})}

\noindent Point $s$ must lie above $h_p$ as in $p$ the normal property is not yet violated. Furthermore, the point of curvature $v$ of $\sigma$ at point $q$ must lie below $h_q$. The shortest self-approaching $q$-$s$ path $\pi_{qs}$ must travel through $v$ and because $h_q$ touches this point, there is a normal $h_v$ that is parallel to $h_p$. Point $s$ cannot lie above $h_p$ and below $h_v$ at the same time, hence this construction cannot exist.

\begin{figure}
    \centering
    \includegraphics[page=19]{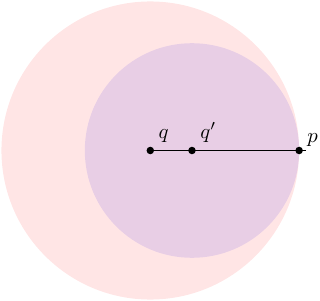}
    \caption{The disc with radius $|pq'|$ centered at point $q'$ which lies on the straight line segment $\overline{pq}$ is fully contained in the disc with radius $|pq|$ centered at $q$.}
    \label{fig:disc_containment}
\end{figure}

\paragraph*{Case 11: \bm{$q \in \sigma(p, t)$}, \bm{$e$} is a straight line segment, and \bm{$q$} a vertex of \bm{$P$} (Fig.~\ref{fig:p_straight_q_straight})}

\noindent The normal $h_p$ touches $\sigma(p, t)$ at point $q$ which is a vertex of $P$. 
The shortest self-approaching $p$-$t$ path $\pi_{pt}$ must intersect or touch $h_p$ at least once in point $q'$ that lies between $p$ and $q$. 
By definition we know that $|pq'| \geq |vq'|$ for every point $v$ of $\pi_{pt}(p, q')$.
Therefore, $|pq| \geq |vq|$, i.e.\ every point of the subpath $\pi_{pt}(p, q')$ lies in a disc centered at $q$ with radius $|pq|$ (Fig.~\ref{fig:disc_containment}).
Thus, the straight line segment $e$ is not part of $\pi_{pt}$ and $e$ is a part of $\sigma$ due to a boundary $d$ of $\mathcal{D}_s$.
Segment $e$ can only be part of $\sigma$ if there is a point $w \in d \subset \sigma$ with a normal $h_w$ that is parallel to $h_p$ and $\sigma(p, w)$ only contains straight line segments and boundaries of $\mathcal{D}_s$.
Otherwise, at the end of $e$ there is either a bend point where $\sigma$ turns to the right or $e$ is tangent to a boundary of $\mathcal{D}_t$. 
In either option, no self-approaching $p$-$t$ path can exist.
Since $w$ lies on a boundary of $\mathcal{D}_s$, there must be a point $w'$ on the shortest self-approaching $w$-$s$ path $\pi_{ws}$ touched by $h_w$.
Therefore, $\sigma(s, w)$ must intersect or touch $h_w$ too. However, $\sigma(s,w)$ cannot intersect $h_w$ as in $h_p$ the normal property is not violated, thus $\sigma$ cannot exist.

\begin{figure}
\centering
\begin{minipage}[t]{.45\textwidth}
  \centering
  \includegraphics[page=13]{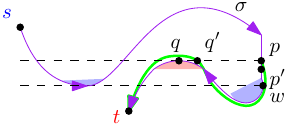}
  \caption{Point $q \in \sigma(p, t)$, $p$ lies on a straight line segment (orange), and $q$ lies on a boundary of $\mathcal{D}_t$ (red). The green path is the shortest self-approaching $p$-$t$ path.}
  \label{fig:p_straight_q_dt}
\end{minipage}%
\hfill
\begin{minipage}[t]{.45\textwidth}
  \centering
  \includegraphics[page=14]{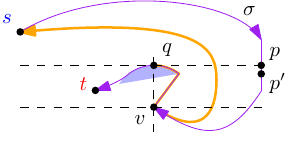}
  \caption{Point $q \in \sigma(p, t)$, $p$ lies on a straight line segment (orange), and $q$ lies on a boundary of $\mathcal{D}_s$ (blue). The orange path is the shortest self-approaching $q$-$s$ path.}
  \label{fig:p_straight_q_ds}
\end{minipage}
\end{figure}

\paragraph*{Case 12: \bm{$q \in \sigma(p, t)$}, \bm{$e$} is a straight line segment, and \bm{$f$} a boundary of \bm{$\mathcal{D}_t$} (Fig.~\ref{fig:p_straight_q_dt})}

\noindent This case is analogous to Case 11 (Fig.~\ref{fig:p_straight_q_straight}).

\paragraph*{Case 13: \bm{$q \in \sigma(p, t)$}, \bm{$e$} is a straight line segment, and \bm{$f$} a boundary of \bm{$\mathcal{D}_s$} (Fig.~\ref{fig:p_straight_q_ds})}

\noindent Point $s$ must lie above $h_p$ as in $p$ the normal property is not yet violated. Furthermore, the center of curvature $v$ of $\sigma$ at point $q$ must lie below $h_p$. The shortest self-approaching $q$-$s$ path $\pi_{qs}$ must travel through $v$ and because $h_q$ touches this point, there is a normal $h_v$ that is parallel to $h_p$. Point $s$ cannot lie above $h_p$ and below $h_v$ at the same time, hence this construction cannot exist.

\begin{figure}
\centering
\begin{subfigure}[t]{.45\textwidth}
  \centering
  \includegraphics[page=15]{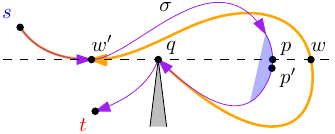}
  \caption{The point of curvature of $\sigma$ at $p$ lies in the direction of $q$. The orange path is the shortest self-approaching $q$-$s$ path.}
  \label{fig:p_ds_q_straight_left}
\end{subfigure}%
\hfill
\begin{subfigure}[t]{.45\textwidth}
  \centering
  \includegraphics[page=18]{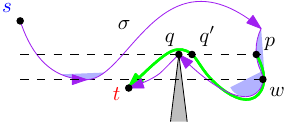}
    \caption{Point of curvature of $\sigma$ at $p$ lies in the opposite direction of $q$. The green path is the shortest self-approaching $p$-$t$ path.}
    \label{fig:p_ds_q_straight_right}
\end{subfigure}%
\caption{Point $q \in \sigma(p, t)$, $q$ is a vertex of $P$, and $p$ lies on a boundary of $\mathcal{D}_s$ (blue).}
\label{fig:p_ds_q_straight}
\end{figure}

\paragraph*{Case 14: \bm{$q \in \sigma(p, t)$}, \bm{$e$} is a boundary of \bm{$\mathcal{D}_s$}, and \bm{$q$} is a vertex of \bm{$P$} (Fig.~\ref{fig:p_ds_q_straight})}

\noindent There are two cases that need to be considered, either the point $v$ of $\sigma$ at point $p$ lies in the direction of $q$ with respect to $p$ or not.
First we consider the case where $v$ lies in the direction of $q$ (Fig.~\ref{fig:p_ds_q_straight_left}).
Since $q$ lies in the negative half-plane $h_{p'}^-$, the center of curvature of $\sigma$ at point $p$ must lie to the left of $q$. 
Therefore, the shortest self-approaching $q$-$s$ path $\pi_{qs}$ must first intersect or touch $h_p$ to the right of $q$ at point $w$ and later to the left at point $w'$. Thus, $|qw'| < |ww'|$ which contradicts the self-approaching property of $\pi_{vs}$.

The case where $v$ does not lie in the direction of $q$ with respect to $p$ (Fig.~\ref{fig:p_ds_q_straight_right}) is analogous to Case 11 (Fig.~\ref{fig:p_straight_q_straight}).

\begin{figure}
\centering
\begin{minipage}[t]{.45\textwidth}
  \centering
  \includegraphics[page=16]{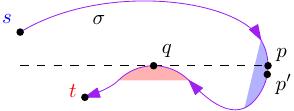}
  \caption{Point $q \in \sigma(p, t)$, $q$ lies on a boundary of $\mathcal{D}_t$ (red), and $p$ lies on a boundary of $\mathcal{D}_s$ (blue).}
  \label{fig:p_ds_q_dt}
\end{minipage}%
\hfill
\begin{minipage}[t]{.45\textwidth}
  \centering
  \includegraphics[page=17]{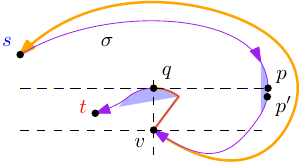}
  \caption{Point $q \in \sigma(p, t)$, $q$ lies on a boundary of $\mathcal{D}_s$ (blue), and $p$ lies on a boundary of $\mathcal{D}_s$ (blue). The orange path is the shortest self-approaching $q$-$s$ path.}
  \label{fig:p_ds_q_ds}
\end{minipage}
\end{figure}

\paragraph*{Case 15: \bm{$q \in \sigma(p, t)$}, \bm{$e$} is a boundary of \bm{$\mathcal{D}_s$}, and \bm{$f$} is a boundary of \bm{$\mathcal{D}_t$} (Fig.~\ref{fig:p_ds_q_dt})}

\noindent This case is analogous to Case 14 (Fig.~\ref{fig:p_ds_q_straight}).

\paragraph*{Case 16: \bm{$q \in \sigma(p, t)$}, \bm{$e$} is a boundary of \bm{$\mathcal{D}_s$}, and \bm{$f$} is a boundary of \bm{$\mathcal{D}_s$} (Fig.~\ref{fig:p_ds_q_ds})}

\noindent Point $s$ must lie above $h_p$ as in $p$ the normal property is not yet violated. Furthermore, the point of curvature $v$ of $\sigma$ at point $q$ must lie below $h_q$. The shortest self-approaching $q$-$s$ path $\pi_{qs}$ must travel through $v$ and because $h_q$ touches this point, there is a normal $h_v$ that is parallel to $h_p$. Point $s$ cannot lie above $h_p$ and below $h_v$ at the same time, hence this construction cannot exist.
\end{proof}
\begin{figure}[th]
     \centering
     \includegraphics[page=12]{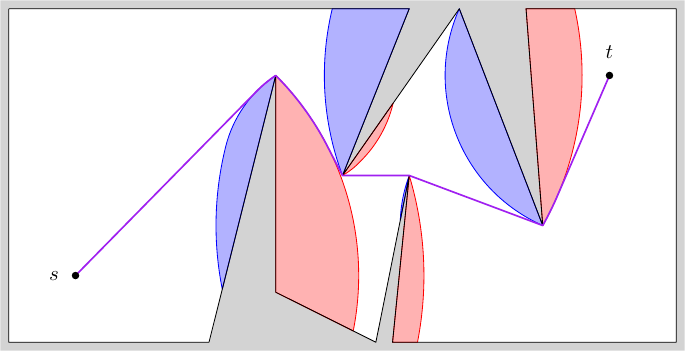}
     \caption{Polygon $P$ with dead regions $\mathcal{D}_s$ (red) and $\mathcal{D}_t$ (blue), the geodesic (purple) of the remaining area is the shortest path with increasing chords.}
     \label{fig:geo_around_dead}
\end{figure}

Using Theorem \ref{thm:geodesic_dead}, the shortest $s$-$t$ path with increasing chords in a polygon $P$ can be found by subtracting the dead regions $\mathcal{D}_t$ and $\mathcal{D}_s$ from polygon $P$ (Fig.~\ref{fig:geo_around_dead}). Therefore, this shortest path with increasing chords can be found in a similar manner to finding the shortest self-approaching $s$-$t$ path as described by Bose et al. \cite{self-approaching_simple_polygons}. As is the case for the algorithm for the shortest self-approaching path, the shortest path with increasing chords can only be found if we assume that transcendental equations can be solved or that an approximation of the path will be calculated.

\section{Conclusions}

In this paper, we showed that the shortest $s$-$t$ path with increasing chords in a simple polygon is unique. 
Furthermore, we showed that the shortest $s$-$t$ path in a polygon minus the dead regions of $s$ and $t$ is the shortest path with increasing chords. 
Therefore, the algorithm for finding the shortest path with increasing chords is similar to finding a shortest self-approaching path.

A natural direction for future research is the question of whether an efficient algorithm exists that can find a shortest $s$-$t$ path with increasing chords in a polygon with holes. The possibility remains open that this problem is NP-hard.

\bibliography{references}

\newpage

\end{document}